\documentclass[a4paper,UKenglish]{lipics_arxiv}

\usepackage{microtype}
\usepackage{amsmath}
\usepackage{amssymb}
\usepackage{gensymb}
\usepackage{xspace}
\usepackage{enumitem}

\graphicspath{{./figures/}}

\title{Stability analysis of kinetic orientation-based shape descriptors}

\titlerunning{Stability analysis of kinetic orientation-based shape descriptors}

\author[1]{Wouter Meulemans}
\author[1]{Kevin Verbeek}
\author[1]{Jules Wulms}
\affil[1]{Dept. of Mathematics and Computer Science, {TU Eindhoven, The Netherlands}\\
\texttt{[w.meulemans|k.a.b.verbeek|j.j.h.m.wulms]@tue.nl}}

\authorrunning{W. Meulemans, K. Verbeek and J. Wulms}

\Copyright{Wouter Meulemans, Kevin Verbeek, Jules Wulms}

\subjclass{Theory of computation $\rightarrow$ Computational geometry}


\newcommand{\reals}{\mathbb{R}}

\newcommand{\orientationspace}{O}

\DeclareMathOperator{\OPT}{OPT}

\DeclareMathOperator{\TS}{\rho_{TS}}
\DeclareMathOperator{\LS}{\rho_{LS}}

\newcommand{\pc}{\textsc{pc}\xspace}
\newcommand{\obb}{\textsc{obb}\xspace}
\newcommand{\strip}{\textsc{strip}\xspace}

\newcommand{\Dt}{\Delta t}

\newcommand*{\etal}{et al.\xspace}

\newcommand{\comment}[3]{\marginpar{\footnotesize\textcolor{#3}{#1:} #2}}
\renewcommand{\comment}[3]{}



\begin{document}

\maketitle

\begin{abstract}
We study three \emph{orientation-based} shape descriptors on a set of continuously moving points: the first principal component, the smallest oriented bounding box and the thinnest strip. Each of these shape descriptors essentially defines a cost capturing the quality of the descriptor and uses the orientation that minimizes the cost. This optimal orientation may be very unstable as the points are moving, which is undesirable in many practical scenarios. If we bound the speed
with which the orientation of the descriptor may change, this may lower the quality of the resulting shape descriptor. In this paper we study the trade-off between stability and quality of these shape descriptors. 

We first show that there is no \emph{stateless algorithm}, an algorithm that keeps no state over time, that both approximates the minimum cost of a shape descriptor and achieves continuous motion for the shape descriptor. On the other hand, if we can use the previous state of the shape descriptor to compute the new state, we can define ``chasing'' algorithms that attempt to follow the optimal orientation with bounded speed. We show that, under mild conditions, chasing algorithms with sufficient bounded speed approximate the optimal cost at all times for oriented bounding boxes and strips. The analysis of such chasing algorithms is challenging and has received little attention in literature, hence we believe that our methods used in this analysis are of independent interest.
\end{abstract}

\newpage

\section{Introduction}
\label{sec:intro}

Given the amount of data that is widely available nowadays, algorithms play an important part in analyzing this data and finding useful patterns in it. For many applications it is important that an algorithm is \emph{stable}: small changes in the input lead to small changes in the output. 
These applications include the analysis or visualization of time-varying data. 
We can effectively visualize how a large set of moving points evolves, for example, by drawing a glyph capturing the direction of movement (e.g. an arrow or line segment) for a few subsets of the points, which gives a clear and comprehensible overview of the data. 
We can also summarize a set of moving points in a single dimension by projecting all points to their first principal component, as used in the visualization technique MotionRugs~\cite{wulms2019spatially}. 
However, if the orientation of the glyphs or first principal component changes erratically even with small changes in the data, it becomes hard to see how the data changes over time. 
In other words, unstable algorithms may result in large visual changes even for small data changes, severely limiting the efficacy of such methods for visualization of time-varying data.
Thus, it is important to develop stable algorithms which deal with these discrete changes of the input in an elegant way. As a result, stable algorithms can efficiently work with the available data without losing their effectiveness over time.

The applications mentioned above all use some type of shape descriptor. Shape descriptors are simplified representations of more complex shapes. They are used, for example, as summaries of a large collection of data, where we are not interested in all the details, but simply want to have an overview of the most important features. Shape descriptors play an important role in many fields that perform shape analysis, such as 
computer vision (shape recognition) \cite{DBLP:journals/pami/BelongieMP02,DBLP:conf/cvpr/BronsteinK10,zhong2009intrinsic}, 
computer graphics (bounding boxes for broad-phase collision detection) \cite{DBLP:journals/cgf/BarequetCGMT96,DBLP:journals/jgtools/Bergen97,DBLP:conf/siggraph/GottschalkLM96,DBLP:journals/tvcg/KlosowskiHMSZ98},
medical imaging (diagnosis or surgical planning) \cite{brzakovic1990approach,DBLP:journals/tmi/GuWCTY04,DBLP:journals/tmi/KelemenSG99,DBLP:journals/ivc/ZitovaF03}, 
and machine learning (shape classification) \cite{DBLP:conf/iccv/SuMKL15,DBLP:conf/iccv/VarmaR07,DBLP:journals/pami/XieDZWF17,DBLP:conf/cvpr/ZhangBMM06}.
In this paper we are interested in shape descriptors that specifically capture the orientation of the data -- in our case a point set. 
The three shape descriptors we consider here are the \emph{first principal component} (\pc), the smallest \emph{oriented bounding box} (\obb) and the thinnest \emph{covering strip} (\strip).
These three shape descriptors can be used in applications (such as the orientation glyphs or MotionRugs) where we always want to output a reasonable orientation, even when this orientation is not very pronounced. Besides all three being orientation-based shape descriptors, all three are unfortunately also unstable: small changes in the point cloud they are representing can result in discrete ``flips'' in their orientation (see Figure~\ref{fig:flips}).
We analyze the stability of these shape descriptors and develop stable variations on them.

\begin{figure}[b]
    \centering
    \includegraphics{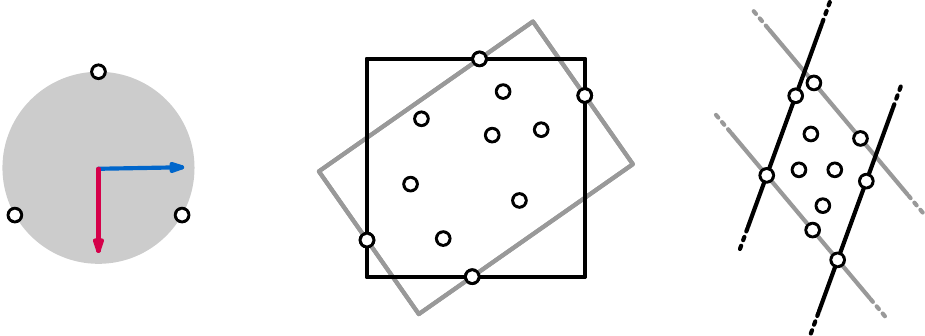}
    \caption{A flip in orientation for \pc, \obb{} and \strip{}. Small changes in the positions of the points make one of the orientations optimal over the other.}
    \label{fig:flips}
\end{figure}

\subparagraph{Problem description.}
The main goal of this paper is to formally analyze the trade-off between quality and stability for orientation-based shape descriptors. Our input consists of a set of $n$ moving points $P = P(t) = \{p_1(t), \ldots, p_n(t)\}$ in two dimensions, where each $p_i(t)$ is a continuous function $p_i\colon [0, T] \rightarrow \reals^2$. We assume that, at each time $t$, not all points are at the same position. 
The output consists of an orientation $\alpha = \alpha(t)$ of the shape descriptor for every point in time $t \in [0, T]$, which need not match the optimal orientation due to stability constraints. To quantify the quality of any output orientation, we define each of the three shape descriptors as the minimum of an optimization function $f(\alpha, P)$. Let $\mathcal{L}(\alpha)$ be the set of lines with orientation $\alpha$, and let $d(L, p)$ be the distance between a point $p$ and a line $L$. Furthermore, let the \emph{extent} of $P$ along a unit vector $\alpha$ be $w_\alpha(P) = \max_{p, q \in P} (p - q) \cdot \alpha$, and let $\alpha^\perp$ be the vector (orientation) orthogonal to $\alpha$. We can define the orientation of \pc, \obb{} or \strip{} as the orientation $\alpha$ that minimizes the following functions:
\begin{description}
\item[Principal component:] $f_{\pc}(\alpha, P) = \min_{L \in \mathcal{L}(\alpha)} \sum_{p \in P} d(L, p)^2$
\item[Oriented bounding box:] $f_{\obb}(\alpha, P) = w_\alpha(P) w_{\alpha^\perp}(P)$
\item[Covering strip:] $f_{\strip}(\alpha, P) = w_{\alpha^\perp}(P)$
\end{description}

\noindent Though various other options are possible, we believe the functions above naturally fit to the shape descriptors. These functions quantify the quality of a shape descriptor for any orientation $\alpha$. This allows us to also consider shape descriptors of suboptimal quality, which do not fully minimize the optimization function. This in turn enables us to make a trade-off between quality and stability. When computing a shape descriptor, we typically compute more than just an orientation. 
However, 
the stability is mostly affected by the optimal orientation: if the optimal orientation changes continuously and the points move continuously, then the shape descriptor changes continuously as well. We therefore ignore other aspects of the shape descriptors to analyze their stability, and assume that these aspects are chosen optimally for the given orientation without any cost with regard to the stability.

Note that an orientation $\alpha(t)$ is an element of the real projective line $\reals\mathbb{P}^1$, but we typically represent $\alpha(t)$ by a unit vector in $\reals^2$ and implicitly identify opposite vectors, which is equivalent.
Furthermore, we assume that the output $\alpha(t)$ is computed for all real values $t \in [0, T]$. This assumption simplifies our analysis. In practice, algorithms can be executed only finitely often, once per some defined time step. Introducing such discreteness into the problem may lead to interesting effects, but it is not the focus of this paper. For sufficiently small time steps, the continuous analysis will provide a good approximation. 


\subparagraph{Kinetic algorithms.}
Algorithms for kinetic (moving) input can adhere to different models, which may influence the results of the stability analysis. Let $\mathcal{A}$ be an algorithm mapping input to output, and $I(t)$ the input depending on time $t$. We distinguish the following models.

\begin{description}
    \item[Stateless algorithms:] The output depends only on the input $I(t)$ at a particular point in time, and no other information of earlier time steps. This in particular means that if $I(t_1) = I(t_2)$, then $\mathcal{A}$ produces the same output at time $t_1$ and at time $t_2$.
    \item[State-aware algorithms:] The algorithm $\mathcal{A}$ has access not only to the input $I(t)$ at a particular time, but also maintains a state $S$ over time; in practice this is typically the output of the previous time step. Thus, even if $I(t_1) = I(t_2)$, then $\mathcal{A}$ may produce different results at time $t_1$ and $t_2$ if the states at those times are different.
    \item[Clairvoyant algorithms:] The algorithm $\mathcal{A}$ has access to the complete function $I(t)$ and can adapt to future inputs. Thus, the complete output over time can be computed offline.
\end{description}

In this paper we consider both stateless and state-aware algorithms, but no clairvoyant algorithms. A stateless algorithm can be stable only if it defines a mapping from input to output that is naturally continuous. On the other hand, a state-aware algorithm can easily enforce continuity by using its state to keep track of earlier output. We can then define a special kind of state-aware algorithm called a \emph{chasing algorithm}:

\begin{description}
    \item[Chasing algorithms:] The algorithm $\mathcal{A}$ maintains the most recent output (in our case a shape descriptor) as the state $S$ and uses only the input at the current time step $I(t)$ to find a new optimal solution $\OPT$. The algorithm then moves the solution in $S$ towards the new optimal solution $\OPT$ at the maximum allowed speed. This solution is the output for the current time step $t$ and is subsequently stored in $S$ for the next time step.
\end{description}

The output at time $t$ might not coincide with $\OPT$, as the solution in $S$ was too different from $\OPT$ and the maximum allowed speed did not allow $S$ to catch up to $\OPT$: the output of our algorithm chases after the optimal solution. The challenge lies in bounding the ratio between the quality of the solutions used by the algorithm and the (unstable) optimal quality. Since we do not consider discrete time steps in this paper, we simply assume that a chasing algorithm maintains a solution over time, and it can choose the rate of change of the solution (e.g., rotation speed/direction) at every point in time.


\subparagraph{Stability analysis.}
Classical considerations to assess algorithms include quality (e.g., optimality, approximation ratio) and efficiency (in e.g. time or memory). These aspects consider the relation between input and output; a worst-case analysis can typically consider a single input.
Stability instead considers how outputs relate between related inputs; a worst-case scenario is a set or sequence of inputs.
This requires a different analysis approach. Here, we use the framework introduced by Meulemans~\etal~\cite{meulemans2017framework}, restricting their definitions to our setting as described below.

Let $\mathcal{A}$ be an algorithm that takes a point set as input and computes an orientation (where the potential state for state-aware algorithms is implicit). 
For \emph{topological stability}, we aim to analyze how well such an algorithm can perform with respect to the (unstable) optimum, if we restrict its output to move continuously but potentially arbitrarily fast for continuously moving input.
The notion of continuity is captured by choosing an appropriate topology for the input space and the output space.
This is straightforward: the input space is $\reals^{2n}$ and the output space is topologically equivalent to the unit circle which we denote here by $\orientationspace$. We choose these topologies for our analysis. 
The \emph{topological stability ratio} of shape descriptor $\Pi$ is then defined as
\[
\rho_{TS}(\Pi) = \inf_{\mathcal{A}} \sup_{P(t)} \max_{t \in [0, T]} \frac{f_{\Pi}(\mathcal{A}(P(t)), P(t))}{\min_\gamma f_{\Pi}(\gamma, P(t))}
\]
where the supremum is taken over all continuously moving point sets $P(t)$, and the infimum is taken over all algorithms $\mathcal{A}$ for which $\mathcal{A}(P(t))$ is continuous. The \emph{$K$-Lipschitz stability ratio} $\rho_{LS}(\Pi,K)$ is defined almost exactly the same: the only difference is that we now take 
the infimum 
over all algorithms $\mathcal{A}$ for which $\mathcal{A}(P(t))$ is $K$-Lipschitz, that is, the output moves continuously with speed bounded by $K$. We measure the output in radians: the orientation or angle can change with a speed of at most $K$ radians per time unit. Lipschitz stability requires bounded input speed and scale invariance \cite{meulemans2017framework}. Here, we assume that points move with at most unit speed and that their diameter is at least 1 at all times.
Note that we can always appropriately scale to meet these assumptions, if not all points coincide in a single location. The unit diameter assumption is only mild, if points represent actual objects such as fish \cite{wulms2019spatially} which have a certain size and cannot occupy the same physical space. 

\subparagraph{Related work.}
Shape descriptors are a wide topic, studied in various subfields of computer science. A full exposition of literature is out of scope. We focus on results related to \pc, \obb or \strip, and results on stability.

The ratio between the volume of the bounding box using the orientation of \pc and the minimal-volume bounding box is unbounded \cite{DBLP:journals/comgeo/DimitrovKKR09}. Similar to \pc, other fitness measures have been considered with respect to oriented lines, such as the sum of distances or vertical distances \cite{DBLP:conf/cocoon/ChinWW99}.
Computing \obb for static point sets is a classic problem in computational geometry. In two dimensions, one side of the optimal box aligns with a side of the convex hull and it can be computed in linear time after finding the convex hull \cite{DBLP:journals/cacm/FreemanS75,toussaint1983solving}; a similar property holds in three dimensions, allowing a cubic-time algorithm \cite{DBLP:journals/ijpp/ORourke85}.
The relevance of bounding boxes in 3D as a component of other algorithms also led to efficient approximation algorithms 
\cite{DBLP:journals/jal/BarequetH01}.
Bounding boxes find applications in tree structures for spatial indexing \cite{DBLP:conf/sigmod/BeckmannKSS90,DBLP:conf/sigmod/Guttman84,DBLP:conf/sigmod/RoussopoulosL85,sellis1987r+} and in speeding up collision detection and ray-tracing techniques \cite{DBLP:journals/cgf/BarequetCGMT96,DBLP:journals/jgtools/Bergen97,DBLP:conf/siggraph/GottschalkLM96}.
The optimal \strip can be computed using the same techniques as \obb in two dimensions  \cite{DBLP:journals/cacm/FreemanS75,toussaint1983solving}.
Agarwal et al. \cite{DBLP:journals/jacm/AgarwalHV04} provide an $O(n +1/\epsilon^{O(1)})$-time approximation algorithm via $\epsilon$-kernels for various measures of a point set, including \strip and \obb.

The stability of kinetic 1-center and 1-median problems has been studied by Bespamyatnikh~\etal~\cite{DBLP:conf/dialm/BespamyatnikhBKS00}. They developed approximations by fixing the speed at which the center/median point moves and provide some results on the trade-off between solution quality and speed.
Durocher and Kirkpatrick \cite{durocher2006steiner,durocher2008bounded} studied the stability of the kinetic 2-center problem, that is, covering a set of moving points with two disks; their approach allows a trade-off between solution quality and the speed at which solution changes. 
De Berg~\etal~\cite{de2013kinetic} show similar results in the black-box kinetic-data-structures model.
Letscher et al. \cite{letscher2016stability} show that the medial axis for a union of disks changes continuously under certain conditions or if it is pruned appropriately . 
Meulemans et al. \cite{meulemans2017framework} introduced a framework for analyzing stability  and apply it to Euclidean minimum spanning trees on a set of moving points. They show bounded topological stability ratios for various topology definitions on the space of spanning trees, and that the Lipschitz stability ratio is at most linear, but also at least linear if the allowed speed for the changes in the tree is too low.
Van der Hoog et al. \cite{hoog2018topological} study the topological stability of the $k$-center problem, showing upper and lower bounds on the stability ratio for various measures.
They also provide a clairvoyant algorithm to determine the best ratio attainable for a given set of moving points.

\subparagraph{Results and organization.}
In Section~\ref{sec:stateless} we prove that there exists no stateless algorithm for any of the shape descriptors that is both topologically stable and achieves a bounded approximation ratio for the quality of the optimal shape descriptor. We then consider state-aware algorithms and analyze the topological stability for each of the shape descriptors in Section~\ref{sec:topological}. In Section~\ref{sec:lipschitz} we analyze the Lipschitz stability of chasing algorithms for all three shape descriptors. We show that chasing algorithms with sufficient speed can achieve a constant approximation ratio for \obb and \strip, if we indeed assume unit speed for the points and that the diameter of the point set is always at least 1. To the best of our knowledge, this is the first time a chasing algorithm has been analyzed. We believe that the new methods that we developed for this challenging stability analysis are of independent interest, and may be applied to the analysis of other chasing algorithms. Section~\ref{sec:lipschitzpc} shortly discusses why we cannot achieve Lipschitz stability for \pc via our chasing algorithm. We conclude the paper in Section~\ref{sec:conclusion}.

\section{Stateless algorithms}
\label{sec:stateless}
We prove that stateless algorithms cannot achieve bounded topological stability ratio for any of the three shape descriptors. This readily implies an unbounded $K$-Lipschitz stability ratio for any $K$.
The argument is entirely topological. A stateless topologically stable algorithm (with output behaving continuously) is a continuous map from the input space to the output space. Important for the theorem below is that, if all points of set $P$ lie on a single line with orientation $\alpha$, then $f_{\Pi}(\gamma, P) = 0$ if and only if $\gamma = \alpha$, for all considered shape descriptors.

\begin{theorem}\label{thm:stateless}
For stateless algorithms $\rho_{TS}(\obb) = \rho_{TS}(\pc) = \rho_{TS}(\strip) = \infty$ if the point set contains at least three points.  
\end{theorem}
\begin{proof}
The idea is to construct a continuous map $j\colon D^2 \rightarrow \reals^{2n}$ on the two-dimensional closed disk, such that the image of $j$ consists of valid point sets (not all points at the same position), and that the image of $\partial D^2$ under $j$ forces the orientation of the shape descriptor. We parameterize $D^2$ using polar coordinates $(r, \phi)$ for $0 \leq \phi < 2 \pi$ and $0 \leq r \leq 1$. We first construct a map $j'$ as follows:
\[
j'(r, \phi) = \{(\frac{r i}{n} \sin\phi, \frac{r i}{n} \cos\phi) \mid 1 \leq i \leq n\}    
\]
Since $j'$ always places all points on a line for a given $\phi$, the orientation of the shape descriptor is always forced (otherwise the approximation ratio is $\infty$). However, $j'(0, \phi)$ is not a valid point set, since it places all points at the origin. Now let $P^{*}$ be a point set of at least three points that are not all collinear. Interpreting $P^*$ as a vector in $\reals^{2n}$), we define $j(r, \phi) = j'(r, \phi) + (1 - r)P^{*}$. By the choice of $P^*$, $j(r, \phi)$ is always a valid point set. Furthermore, the orientation of the shape descriptors is still fixed for point sets $j(1, \phi) = j'(1,\phi)$, namely $\alpha = \phi \pmod{\pi}$. As a result, any stateless algorithm with an approximation ratio $\rho < \infty$ defines, along with $j$, a continuous mapping $h$ from $D^2$ to unit circle $\orientationspace$ where $\partial D^2$ is mapped twice around $\orientationspace$, i.e., a double cover of $\orientationspace$. The continuous mapping $h$ is simply the output of the stateless algorithm composed with $j$. We claim that mapping $h$ from $D^2$ to $\orientationspace$ cannot exist.

For the sake of contradiction, assume that such a map $h$ exists.
Let $f,g: \orientationspace \to D^2$ be continuous functions. Function $f$ maps every point $x \in \orientationspace$ to the boundary $\partial D^2 \subseteq D^2$ such that the mapping covers the whole boundary once, while $g$ maps all $x\in \orientationspace$ to a single point $y\in D^2$. We can continuously shrink the image of $f$ to a single point in $D^2$, in particular the image of $g$; hence $f$ and $g$ are homotopic.
We now consider $h \circ f$ and $h\circ g$ and use the degree of these mappings (as first defined in \cite{brouwer1911abbildung}) to show that $h$ cannot exist.
Since $f$ maps $\orientationspace$ to the boundary $\partial D^2 \subseteq D^2$, and $h$ maps $\partial D^2$ to a double cover of $\orientationspace$, we know that the degree of $h\circ f$ is two. On the other hand, $g$ maps all of $\orientationspace$ to a single point in $D^2$, therefore $h\circ g$ has degree 0. By the Hopf theorem~\cite{milnor1997topology} $h\circ f$ and $h\circ g$ cannot be homotopic, as they can only be homotopic if and only if their degrees are equal. However, $h\circ f$ and $h\circ g$ must be homotopic, since homotopy equivalence is compatible with function composition and $f$ is homotopic to $g$. This contradiction implies that $h$ cannot exist.
Thus the topological stability ratio is $\infty$ for stateless algorithms.
\end{proof}

\section{Topological stability}
\label{sec:topological}

In this section we turn to state-aware algorithms, and we analyze the topological stability of the shape descriptors. Specifically, we prove the following results.
\begin{theorem}\label{thm:topologicalstab}
The topological stability ratios of the shape descriptors are:
\begin{itemize}[noitemsep,topsep=0pt]
\item $\TS(\pc) = 1$,
\item $\TS(\obb) = \frac{5}{4}$,
\item $\TS(\strip) = \sqrt{2}$.
\end{itemize}
\end{theorem}
To prove an upper bound on the topological stability ratio, we need an algorithm that produces an output that changes continuously, but may change arbitrarily fast. Our algorithm works as follows. Whenever an optimal solution undergoes a discrete change, the output moves continuously between the solution before and after the discrete change. When analyzing the topological stability ratio, we consider all intermediate solutions that the continuous output goes through. It is therefore sufficient to consider a single point in time where the optimal solution undergoes a discrete change to find an upper bound, since that is where the output differs most from an optimal solution. On the other hand, to prove a lower bound on the topological stability ratio, we must construct a full time-varying point set such that the corresponding approximation ratio must occur at some point in time during this motion, regardless of which algorithm is used. Note that a lower bound on the topological stability ratio is immediately also a lower bound on the $K$-Lipschitz stability ratio for any value of $K$.

\begin{figure}[ht]
    \centering
    \includegraphics{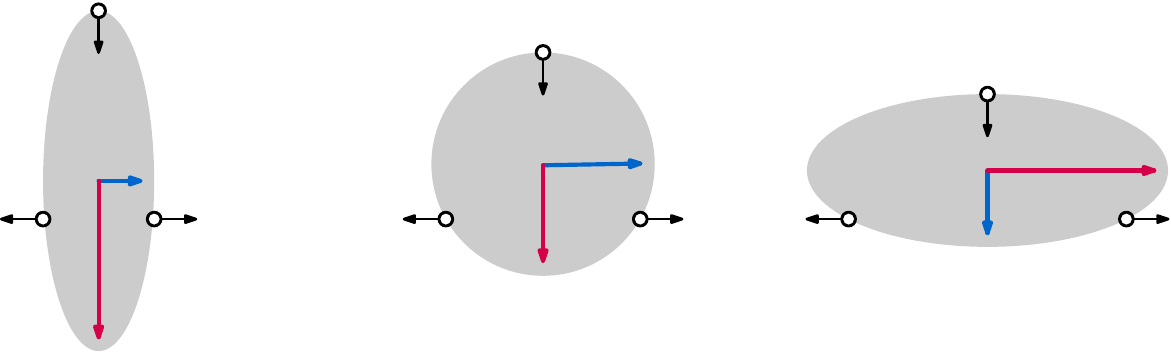}
    \caption{Moving points causing a flip between the first principal component (red) and second principal component (blue).}
    \label{fig:pc-topological}
\end{figure}

\begin{lemma}\label{lem:pc-topological}
$\TS(\pc) = 1$
\end{lemma}
\begin{proof}
Consider a time $t$ where the first principal component flips between two orientations, represented by unit vectors $\vec{v}_1$ and $\vec{v}_2$. The first principal component is the orientation of the line that minimizes the sum of squared distances between the points and the line, as described in Section~\ref{sec:intro}. It can be computed by centering the point set at the mean of the coordinates, computing the covariance matrix of the resulting point coordinates, and extracting the eigenvector of this matrix with the largest eigenvalue. Since eigenvalues change continuously if the data changes continuously~\cite[Theorem 3.9.1]{tyrtyshnikov2012brief}, both $\vec{v}_1$ and $\vec{v}_2$ must have some eigenvalue $\lambda^{*}$ at time $t$. But that means that every interpolated vector $\vec{v} = (1 - u) \vec{v}_1 + u \vec{v}_2$ also has eigenvalue $\lambda^{*}$, since $C \vec{v} = (1 - u) C \vec{v}_1 + u C \vec{v}_2 = (1 - u) \lambda^{*} \vec{v}_1 + u \lambda^{*}\vec{v}_2 = \lambda^{*} \vec{v}$. As a result, $f_\pc(\vec{v}) = f_\pc(\vec{v_1})$, and we can continuously change orientation from $\vec{v}_1$ to $\vec{v}_2$ without decreasing the quality of the shape descriptor.
\end{proof}

\begin{lemma}\label{lem:obb-topological-UB}
$\TS(\obb) \leq \frac{5}{4}$
\end{lemma}
\begin{proof}
Consider a time $t$ at which two distinct oriented bounding boxes $A$ and $B$ have minimum area; both are assumed to have the smallest area $1$ without loss of generality, as the problem at hand is invariant under scaling. 
At this time $t$ we continuously change the orientation of the box between that of $A$ and $B$ while making sure that the box still contains all points (see Figure~\ref{fig:obb-topological-UB}). The goal is to compute the maximal size of the intermediate box in the worst case. 
Note that we may rotate either clockwise or counterclockwise; we always choose the direction that minimizes this maximal intermediate size.

Let $a$ and $b$ denote the length of the major axes of $A$ and $B$ respectively. 
Let angle $\alpha$ denote the smallest angle between the orientations of the major axes. Note that $\alpha \in \{0, \pi/2\}$ leads to $A$ and $B$ being identical, and that our problem is invariant under rotation, reflection and translation.
We thus assume without loss of generality:
\begin{itemize}
    \item $b \geq a \geq 1$;
    \item $0 < \alpha < \pi/2$;
    \item $B$ is centered at the origin, and $A$ at $(dx,dy)$;
    \item the major axis of $A$ is horizontal;
    \item $\alpha$ describes a counterclockwise angle from the major axis of $A$ to the major axis of $B$.
\end{itemize}

\begin{figure}[b]
    \centering
    \includegraphics{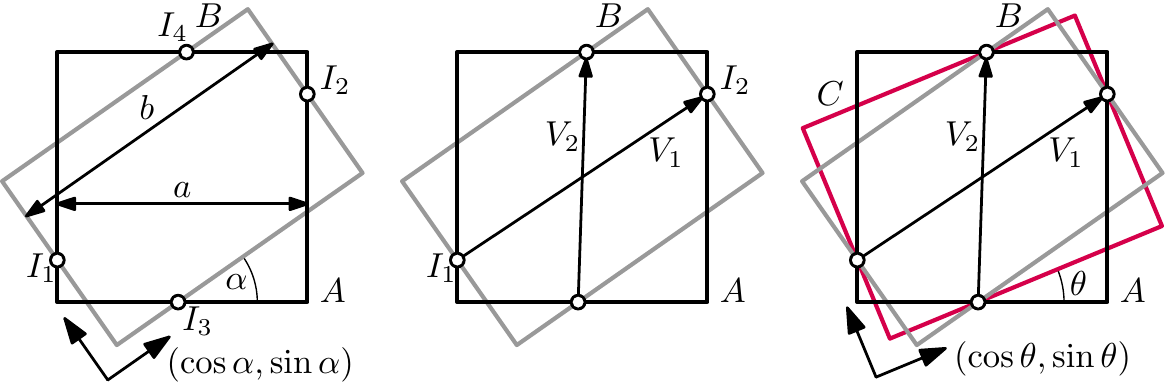}
    \caption{Construction of closed formula for area of intermediate solution C.}
    \label{fig:obb-topological-UB}
\end{figure}

The points $P$ must all be contained in the intersection of $A$ and $B$, for otherwise $A$ and $B$ would not be bounding boxes. Furthermore, no side of $A$ may be completely outside $B$ or vice versa, for otherwise one of the boxes could be made smaller. Thus, all sides intersect, and we are interested in four of these intersections $I_1, \ldots, I_4$ (see Figure~\ref{fig:obb-topological-UB}). Specifically, we want to use the intersections that allow us to derive a valid upper bound on the size of the bounding box during rotation. Since the intersections depend on the direction of rotation, we choose the intersections that allow us to rotate from $A$ to $B$ in counterclockwise direction.
The coordinates of the intersections can easily be computed (see Figure~\ref{fig:obb-topological-UB}). For example, for $I_1=(x_1,y_1)$ we solve for the following two equations: $x_1 = dx - a/2$ and $x_1 \cos{\alpha} + y_1\sin{\alpha} = -b/2$. The resulting coordinates of all intersections are:

\begin{itemize}
    \item $I_1 = (-\frac{a}{2}+dx, -\frac{b-(a-2dx)\cos{\alpha}}{2\sin{\alpha}})$
    \item $I_2 = (\frac{a}{2}+dx, \frac{b-(a+2dx)\cos{\alpha}}{2\sin{\alpha}})$
    \item $I_3 = ((-\frac{1}{2a}+dy)\frac{\cos{\alpha}}{\sin{\alpha}}+\frac{1}{2b\sin{\alpha}} ,-\frac{1}{2a}+dy)$
    \item $I_4 = ((\frac{1}{2a}+dy)\frac{\cos{\alpha}}{\sin{\alpha}}-\frac{1}{2b\sin{\alpha}} ,\frac{1}{2a}+dy)$
\end{itemize}

Now consider an intermediate box $C$ with angle $\theta \leq \alpha$ with respect to box $A$. Note that $C$ contains the intersection of $A$ and $B$ as long as it contains $I_1, \ldots, I_4$. We can define two vectors $V_1 = I_2 - I_1$ and $V_2 = I_4 - I_3$, which we project to lines at angle $\theta$ and $\theta+\frac{\pi}{2}$ to obtain the lengths of the sides of $C$. Note that $V_1$ and $V_2$ depend only on $a$, $b$ and $\alpha$, but not on $dx$ and $dy$. Thus, using $V_1$ and $V_2$ we can obtain a formula for the area of $C$, which we call $\mathcal{C}$. 
\begin{align*}
\mathcal{C}(a, b, \alpha, \theta) &= V_1 \cdot (\cos{\theta}, \sin{\theta}) \times V_2 \cdot (-\sin{\theta}, \cos{\theta}) \\
&= \frac{(b\sin(\alpha-\theta) + a\sin{\theta}) * (a\sin(\alpha-\theta) + b\sin{\theta})}{ ab\sin^2{\alpha}}
\end{align*}

We are now interested in the maximum of $\mathcal{C}$. We start by finding the partial derivative of $\mathcal{C}$ with respect to $\theta$:
\[\frac{\partial \mathcal{C}}{\partial \theta} = \frac{(a^2 + b^2 -2ab\cos{\alpha})\sin(\alpha - 2\theta)}{ab\sin^2{\alpha}}\]
First observe that $\frac{\partial \mathcal{C}}{\partial \theta} = 0$ if and only if $\theta = \alpha/2$ in the chosen domain, which implies that we can set $\theta = \alpha/2$. 
Using double-angle formulas, we may simplify $\mathcal{C}$ in this setting.

\begin{align*}
\mathcal{C}(a, b, \alpha, \alpha/2) &= \frac{(a+b)^2\sin^2(\alpha/2)}{ab\sin^2{\alpha}}\\
&=\frac{(a+b)^2\sin^2(\alpha/2)}{ab(2\sin(\alpha/2)\cos(\alpha/2))^2}\\
&=\frac{(a+b)^2}{2ab(2\cos^2(\alpha/2))}\\
&=\frac{(a+b)^2}{2ab(1+\cos(\alpha))}
\end{align*}

We now split the domain of $\alpha$ into $(0,\frac{\pi}{4}]$ and $(\frac{\pi}{4}, \frac{\pi}{2})$, and prove both cases separately.

In the first case, when $\alpha \in (0,\frac{\pi}{4}]$, let $c \geq 1$ be such that $b = c a$ (since $b \geq a$). 
As $b$ is at most the length of the diagonal of $A$, we get that $c \leq \sqrt{a^2 + (1/a)^2}/a = \sqrt{1 + 1/a^4} \leq \sqrt{2}$ (since $a \geq 1$). The resulting formula is $\mathcal{C}(a, c a, \alpha, \alpha/2) = \frac{(1+c)^2}{2c(1+\cos{\alpha})}$. This function is maximized when $c$ and $\alpha$ are maximized. We thus set $\alpha = \pi/4$ and $c = \sqrt{2}$ to obtain that $\mathcal{C}(1, \sqrt{2}, \pi/4, \pi/8) \leq \frac{1}{2} + \frac{1}{2}\sqrt{2} < \frac{5}{4}$ and thus this case meets the bound claimed.

What remains is to prove the case where $\alpha \in (\frac{\pi}{4}, \frac{\pi}{2})$. It might now be beneficial to rotate $A$ clockwise instead of counterclockwise, to align the minor axis of $A$ with the major axis of $B$: this clockwise rotation may result in smaller intermediate solutions $C$. 
Since $\mathcal{C}$ can only deal with counterclockwise rotation, we have to use different parameters to deal with the described situation. To simulate the clockwise rotation, we use $\mathcal{C}(1/a, b, \pi/2 - \alpha, \theta)$;  this reflects the whole setup over direction $\pi/4$ effectively considering the minor axis as the major axis instead. Note that we did not use the assumption that $a \geq 1$ anywhere above, until within the other case where $\alpha \leq \pi/4$. 
Note that $\frac{\partial \mathcal{C}}{\partial \theta} = 0$ depends on the parameters we fill in, hence we can set $\theta = (\pi/2 - \alpha)/2 = \pi/4 - \alpha/2$ to find a maximum in this case. 

For all possible values of $a$, $b$ and $\alpha$, we need to find the area of the largest intermediate box $C$. Since we can choose whether we rotate clockwise or counterclockwise, we find the area by taking the minimum of $\mathcal{C}(a, b, \alpha, \alpha/2)$ and $\mathcal{C}(1/a, b, \pi/2 - \alpha, \pi/4 - \alpha/2)$. 
We first simplify the latter.

\begin{align*}
\mathcal{C}(1/a, b, \pi/2 - \alpha, \pi/4 - \alpha/2) &= \frac{(\frac{1}{a}+b)^2}{2\frac{1}{a}b(1+\cos(\pi/2 - \alpha))}\\
&= \frac{a^2(\frac{1}{a}+b)^2}{2ab(1+\sin(\alpha))}\\
&= \frac{(1+ab)^2}{2ab(1 + \sin(\alpha))}\\
\end{align*}

To find the maximum of the function, we can use a mathematical program that looks for the values of $a,b$ and $\alpha$ that comply to a set of constraints and maximize a target function. All constraints come from the assumptions, but we add a final constraint similar to what we did in the $\alpha\in(0,\frac{\pi}{4}]$ case. To ensure that all 4 intersection points exist, the projection of the diagonal of $A$ to the major axis of $B$ should be larger than $b$. Hence we add the constraint $b\leq a\cos{\alpha} + \frac{1}{a}\sin{\alpha}$. The mathematical program now looks as follows:

\begin{align*}
    \text{\textbf{maximize}} \quad & \min\left( \frac{(a+b)^2}{2ab(1+\cos(\alpha))}, \frac{(1+ab)^2}{2ab(1 + \sin(\alpha))} \right) \\
    \text{\textbf{subject to}} \quad & 
        \begin{aligned}[t] b &\leq a\cos{\alpha} + \frac{1}{a}\sin{\alpha}\\
        \pi/4 < \alpha &< \pi/2\\
        1 \leq a &\leq b \\
        \end{aligned}
\end{align*}



Using the mathematical program we can verify that the area is at most $\frac{5}{4}$.
\end{proof}

\begin{figure}[ht]
	\centering
    \includegraphics{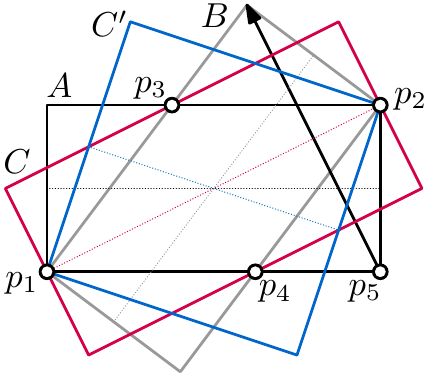}
    \caption{Moving points forcing a flip and a bounding box at least as big as $C$ or $C'$.}
    \label{fig:obb-topological-LB}
\end{figure}

\begin{lemma}\label{lem:obb-topological-LB}
$\TS(\obb) \geq \frac{5}{4}$
\end{lemma}
\begin{proof}
Refer to Figure~\ref{fig:obb-topological-LB} for illustration.
Consider a point set $P$ with four static points $p_1 = (0,0)$, $p_2 = (2,1)$, $p_3 = (0.75,1)$ and $p_4 = (1.25,0)$, and point $p_5$ moving linearly from $(2,0)$ to $(1.2,1.6)$.
The static points allow two minimal bounding boxes of area $2$ and aspect ratio $2$: one with orientation $0$ (box $A$) and one with orientation $2 \arctan(\frac{1}{2}) \approx 53.13$ degrees (box $B$).
As $p_5$ is always in $A$ or $B$, one of these boxes is always optimal. Initially, only $A$ contains $p_5$ and in the end only box $B$ does.
The angles $\arctan(\frac{1}{2})$ (box $C$) and $\pi/4 - \arctan(\frac{1}{2})$ (box $C'$) give an intermediate box of size $2.5 = \frac{5}{4} \cdot 2$ on the static points. 
$C$ is encountered on a counterclockwise rotation from $A$ to $B$, and $C'$ on a clockwise rotation. 
Neither $C$, $C'$, nor any box rotated more towards $B$ contains the initial location of $p_5$. Similarly, neither $C$, $C'$, nor any box rotated more towards $A$ contains the final location.

To derive a contradiction, assume a continuously moving \obb exists that achieves a ratio strictly less that $\frac{5}{4}$.
This ratio implies that initially the \obb orientation is clockwise between $C$ and $C'$, and at the end of the motion it is not.
However, as the assumed \obb moves continuously through $\orientationspace$, it must at some point have been in the orientation of $C$ or $C'$.
But this implies a ratio of $\frac{5}{4}$, contradicting our assumption and proving the lower bound.
\end{proof}

\begin{lemma}\label{lem:strip-topological-UB}
$\TS(\strip) \leq \sqrt{2}$
\end{lemma}
\begin{proof}
Consider a time $t$ at which there are two thinnest strips $A$ and $B$ of width $1$ with different orientations. All points must be contained in the diamond-shaped intersection $D$ of $A$ and $B$ (see Figure~\ref{fig:strip-topological-UB}). If we continuously rotate a strip $C$ from $A$ to $B$, then at some point the width of $C$ must be at least the length of one of the diagonals of $D$. To maximize the length of the shortest diagonal, $D$ must be a square with side length $1$. Therefore, the width of $C$ is at most $\sqrt{2}$ during the rotation from $A$ to $B$.
\begin{figure}[ht]
\begin{minipage}[t]{0.3\linewidth}
    \centering
    \includegraphics{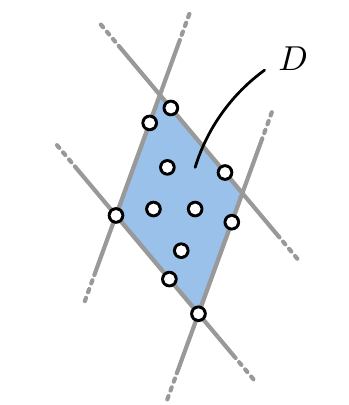}
    \caption{A configuration having two strips of minimal width and overlapping area $D$.}
    \label{fig:strip-topological-UB}
\end{minipage}
\hfill
\begin{minipage}[t]{0.65\linewidth}
    \centering
    \includegraphics{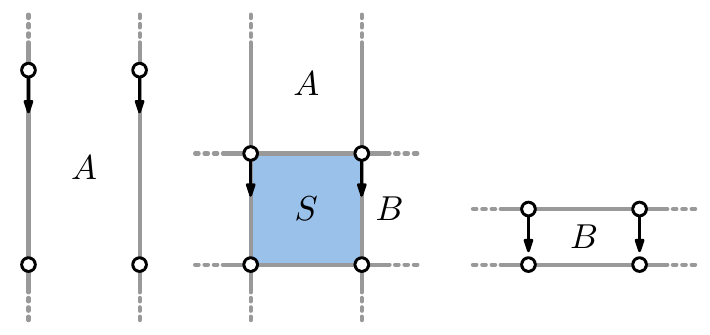}
    \caption{An instance of moving points where the thinnest strip changes orientations. The configuration that leads to the best intermediate solutions is shown in the middle.}
    \label{fig:strip-topological-LB}
\end{minipage}
\end{figure}
\end{proof}

\begin{lemma}\label{lem:strip-topological-LB}
$\TS(\strip) \geq \sqrt{2}$
\end{lemma}
\begin{proof}
Let $P$ consist of four points positioned in a unit square $S$. There are two thinnest strips $A$ and $B$ for $P$, each of which is oriented along a different pair of parallel edges of $S$ (see Figure~\ref{fig:strip-topological-LB}). If the orientation of a strip is $\pi/4$ away from $A$ and $B$, then its width is $\sqrt{2}$.

Now assume that the top points of $S$ are moving along the vertical sides of $S$, starting from high above. Clearly, at the start of this motion, any strip $C$ approximating the thinnest strip must align with $A$. At the end of this motion, when the top points align with the bottom points of $S$, the strip $C$ must align with $B$. Therefore, the strip $C$ must at some point make an angle of $\pi/4$ with $A$ and $B$. If $x$ is the distance between the top and bottom points of $S$, then the width of $C$ at this orientation is $(1 + x) \sqrt{2} / 2$. The minimal width is $\min(x, 1)$. It is easy to verify that this ratio is at least $\sqrt{2}$, which concludes the proof.
\end{proof}

\section{Lipschitz stability}
\label{sec:lipschitz}

To derive meaningful bounds on the Lipschitz stability ratio, we assume the points move with at most unit speed. Furthermore, the relation between distances and speeds in input and output space should be \emph{scale invariant}~\cite{meulemans2017framework}. This is currently not the case: if we scale the coordinates of the points, then the distances in the input space change accordingly, but the distances in the output space (between orientations) do not. To remedy this problem, we require that diameter $D$ of $P(t)$ is at least $1$ for every time $t$.
For \pc this is not sufficient, as we argue in Section~\ref{sec:lipschitzpc}.

To produce a $K$-Lipschitz stable solution we use a chasing algorithm similar to the generic algorithm introduced in Section~\ref{sec:intro}. The algorithm maintains a solution over time, and it can rotate towards the optimal solution at every point in time.
However, there are two differences from the generic algorithm. First, instead of chasing the orientation of an optimal solution \obb, we chase the orientation of the diametrical pair. Although chasing an optimal shape descriptor would be better in general, chasing the diametrical pair is easier to analyze and sufficient to obtain an upper bound on the $K$-Lipschitz stability ratio for \obb. Second, $K$-Lipschitz stability enforces a maximum speed at which the algorithm can move towards the solution we are chasing. This speed bound depends on parameter $K$ and on how quickly the input changes -- faster moving points require faster rotation to achieve the same ratio. Since we assume the points move at unit speed, the $K$-Lipschitz stable solution changes with at most $K$ radians per time step.

\subsection{Chasing the diametrical pair}
We denote the orientation of the diametrical pair as $\alpha = \alpha(t)$ and the diameter as $D = D(t) \geq 1$. Furthermore, let $W = W(t)$ be the width of the thinnest strip with orientation $\alpha(t)$ covering all points in $P(t)$, and let $z = z(t) = W(t) / D(t)$ be the \emph{aspect ratio} of the \emph{diametric box} with orientation $\alpha(t)$.  Finally, our chasing algorithm has orientation $\beta = \beta(t)$ and can change orientation with at most constant speed $K$. We generally omit the dependence on $t$ if $t$ is clear from the context.

\subparagraph{Approach.}
The main goal is to keep chasing orientation $\beta$ as close as possible to orientation $\alpha$ of an optimal diameter box, specifically within a sufficiently small interval around $\alpha$. The challenge lies with the discrete flips of $\alpha$. We must argue that, although flips happen instantaneously, a short time span does not admit many flips over a large angle in the same direction; otherwise we can never keep $\beta$ close to $\alpha$ with a bounded speed. Furthermore, the size of the interval must depend on the aspect ratio $z$, since if $z = 0$, the interval around $\alpha$ must have zero size as well to guarantee a bounded approximation ratio.

For the analysis we introduce three functions depending on $z$: $T(z)$, $H(z)$, and $J(z)$. Function $H(z)$ defines an interval $[\alpha - H(z), \alpha + H(z)]$ called the \emph{safe zone}. We aim to show that, if $\beta$ leaves the safe zone at some time $t$, it must return to the safe zone within the time interval $(t, t + T(z)]$. We also define a larger interval $I = [\alpha - H(z) - J(z), \alpha + H(z) + J(z)]$. We refer to the parts of $I$ outside of the safe zone as the \emph{danger zone}. Figure~\ref{fig:interval} shows $I$ at time $t$ and time $t+\epsilon$. Although $\beta$ may momentarily end up in the danger zone due to discontinuous changes, it must quickly find its way back to the safe zone. We aim to guarantee that $\beta$ stays within $I$ at all times. 
\begin{figure}[ht]
	\centering
    \includegraphics{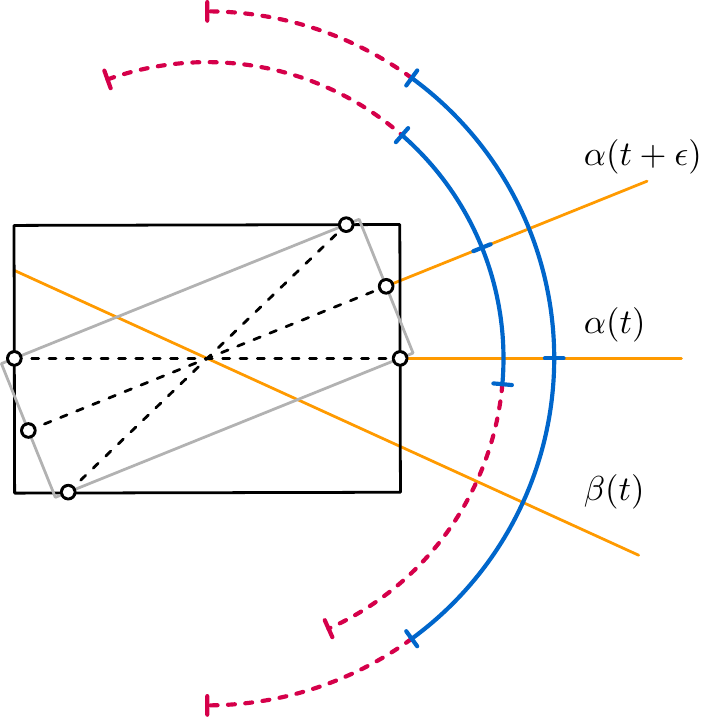}
    \caption{Intervals at time $t$ (outer) and time $t+\epsilon$ (inner). The safe and danger zones are indicated in blue and dashed red respectively; orientations are shown in yellow.}
    \label{fig:interval}
\end{figure}
Let $E = E(t)$ refer to an endpoint of $I$. We call $J(z)$ the \emph{jumping distance} and require that $J(z)$ upper bounds how far $E$ can ``jump'' instantaneously. Let $\Delta E(z, \Delta t)$ denote how far $E$ moves over a time period of length $\Delta t$, starting with a diameter box of aspect ratio $z$. We then require that $\Delta E(z, 0) \leq J(z)$. Note that by defining this upper bound, $J(z)$ is defined recursively through $E$, since by definition $\Delta E(z, \Delta t)$ is upper bounded by how much $\alpha$, $H(z)$ and $J(z)$ change over time $\Delta t$. Hence we need to carefully choose the right function for $J(z)$. For the other functions we choose $T(z) = z/4$ and $H(z) = c \arcsin(z)$ for a constant $c$ (chosen later).

\subparagraph{Changes in orientation and aspect ratio.}
To verify that the chosen functions $T(z)$ and $H(z)$ satisfy the intended requirements, and to define the function $J(z)$, we need to bound how much $\alpha$ and $z$ can change over a time period of length $\Delta t$. We refer to these bounds as $\Delta \alpha(z, \Delta t)$ and $\Delta z(z, \Delta t)$, respectively. Note that, since the diameter can change discontinuously, we generally have that $\Delta \alpha(z, 0) > 0$ and $\Delta z(z, 0) > 0$.

\begin{lemma}\label{lem:flipangle}
$\Delta\alpha(z, \Delta t) \leq \arcsin(z + \Delta t (2 + 2z))$ for $\Delta t\leq (1 - z) / (2 + 2z)$.
\end{lemma}
\begin{proof}
Refer to Figure~\ref{fig:maxangle} for illustrations.
Let $D$ be the diameter at time $t$; the width of the strip containing all points is $z D$. Also, let $D'$ be the diameter at time $t + \Delta t$, and let $(p'_1, p'_2)$ be the diametrical pair at that time, such that the diametrical orientation differs an angle $\gamma$ from the orientation at time $t$. Note that $\Delta t \geq |D - D'|/2$. Both $p'_1$ and $p'_2$ must have been in the diametric box at time $t$. This means that $\Delta t \geq \frac{D'}{2} \sin(\gamma) - \frac{z D}{2}$. As $\Delta t$ is minimized when $D \geq D'$, we can obtain a lower bound for $\Delta t$ by equalizing $(D - D')/2 = \frac{D'}{2} \sin(\gamma) - \frac{z D}{2}$. We obtain that $\Delta t \geq \frac{D' (\sin(\gamma) - z)}{2 + 2z} \geq \frac{\sin(\gamma) - z}{2 + 2z}$. This is equivalent to $\gamma \leq \arcsin(z + \Delta t (2 + 2z))$. Note that this function is only well-defined for $\Delta t \leq \frac{1 - z}{2 + 2z}$.
\begin{figure}[b]
    \centering
    \includegraphics{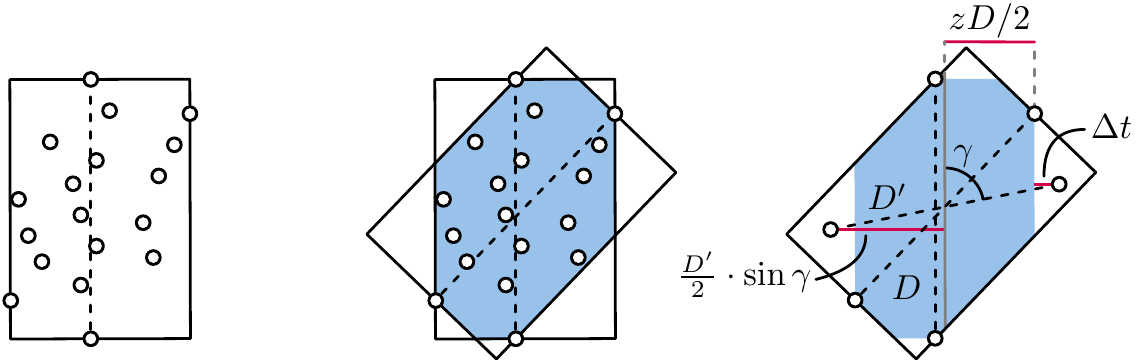}
    \caption{A diametrical box with dimensions $D$, $z D$ containing all points. If the orientation of the diametrical pair changes, all points must lie in the blue area. The orientation can change  further in the same direction, after some points move outside of the blue area and establish a new diameter.}
    \label{fig:maxangle}
\end{figure}
\end{proof}

\begin{lemma}\label{lem:flipaspectratio}
$\Delta z(z, \Delta t) \leq z - \frac{\sin(\frac{1}{2}\arcsin(z)) - 2 \Delta t}{1 + 2 \Delta t}$ for $\Delta t \leq \sin(\frac{1}{2}\arcsin(z))/2$.
\end{lemma}
\begin{proof}
Refer to Figure~\ref{fig:flipaspectratio} for illustrations. 
Let diameter $D$ at time $t$ be realized by the pair of points $(p_1, p_2)$ with orientation $\alpha$. The width of the diametric box is determined by two points $q_1$ and $q_2$; the distance between $q_1$ and $q_2$ is at most $D$. 
To minimize the aspect ratio of the diametric box at time $t + \Delta t$, we need to find a thinnest strip that contains all of $p_1$, $p_2$, $q_1$, and $q_2$. 
For the thinnest strip, all four points are on the boundary of the strip in the worst case, and we assume w.l.o.g. that $p_1$ and $q_1$ are on the same side of the strip (same for $p_2$ and $q_2$). Consider the following lines: $L$ oriented in the orientation of the thinnest strip (parallel to its boundary), $L_p$ spanned by $p_1p_2$ and $L_q$ spanned by $q_1q_2$. Let the angle between $L_p$ and $L_q$ be $\gamma$, $\gamma \geq \arcsin(z)$. 
The distance between $q_1$ and $q_2$ is then $z D / \sin(\gamma)$. 
We denote the angle between $L$ and $L_p$ by $\gamma_p$, and between $L$ and $L_q$ by $\gamma_q$.
We observe that $\gamma_p + \gamma_q = \gamma$, as the orientation of $L$ must bisect the $\gamma$ angle for the strip to be thinnest. 
The width of the strip is $\max(D \sin(\gamma_p), z D \sin(\gamma_q) / \sin(\gamma))$. We show that this width is at least $D \sin(\frac{1}{2}\arcsin(z))$. This is clearly the case if $\gamma_p \geq \frac{1}{2}\arcsin(z)$, so assume the contrary. Since the function $\sin(\gamma - \gamma_p) / \sin(\gamma)$ is increasing, it is optimal to set $\gamma = \arcsin(z)$. But then $z D \sin(\gamma_q) / \sin(\gamma) = D \sin(\gamma_q) > D \sin(\frac{1}{2}\arcsin(z))$. Thus, the width of the thinnest strip is at least $D \sin(\frac{1}{2}\arcsin(z))$.
As a result, the width of the diametric box at time $t + \Delta t$ is at least $D \sin(\frac{1}{2}\arcsin(z)) - 2 \Delta t$. Since the diameter $D' \leq D + 2 \Delta t$, the final aspect ratio is $z' \geq (D \sin(\frac{1}{2}\arcsin(z)) - 2 \Delta t) / (D + 2 \Delta t)$. Since $D \geq 1$, we obtain that $\Delta z(z, \Delta t) \leq z - (\sin(\frac{1}{2}\arcsin(z)) - 2 \Delta t) / (1 + 2 \Delta t)$. Note that this bound is meaningful only for $\Delta t \leq \sin(\frac{1}{2}\arcsin(z))/2$.
\end{proof}
\begin{figure}[ht]
    \centering
    \includegraphics{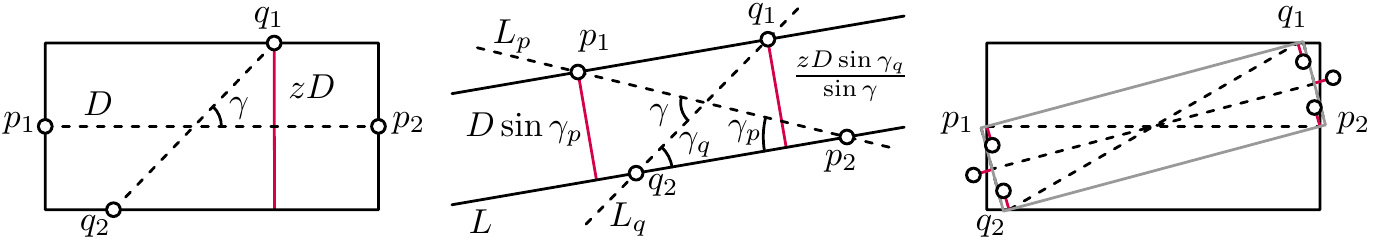}
    \caption{Illustrations for the proof of Lemma \ref{lem:flipaspectratio}. (Left) A diametrical box with aspect ratio $z$ and points $p_1,p_2,q_1$ and $q_2$ at the boundary. (Middle) A strip with points located at the boundary, the width of the strip is the maximum of the red lines. (Right) The smallest diametrical box for time $t+\epsilon$ and in red the distance the points can travel in $\Delta t$ to further shrink the aspect ratio.}
    \label{fig:flipaspectratio}
\end{figure}

\subparagraph{Jumping distance.}
We can now derive a valid function for $J(z)$. Recall that we require that $J(z)$ is at least the amount $\Delta E(z, \Delta t)$ that $E$ can move in $\Delta t = 0$ time. 

\begin{lemma}\label{lem:jumping}
$J(z) = (c+2)\arcsin(z)$ is a valid jumping distance function.
\end{lemma}
\begin{proof}
By Lemma~\ref{lem:flipangle} and Lemma~\ref{lem:flipaspectratio} we get that $\Delta E(z, 0) \leq \Delta \alpha(z, 0) + H(z) - H(z - \Delta z (z, 0)) + J(z) - J(z - \Delta z (z, 0))$. Since $\Delta \alpha(z, 0) \leq \arcsin(z)$ and $\Delta z (z, 0) \leq z - \sin(\frac{1}{2}\arcsin(z))$, we get after simplification that $\Delta E(z, 0) \leq (1 + c/2) \arcsin(z) + J(z) - J(\sin(\frac{1}{2}\arcsin(z)))$. Since we require that $J(z) \geq \Delta E(z, 0)$, it suffices to show that the following holds: $J(\sin(\frac{1}{2}\arcsin(z))) \geq (1 + c/2) \arcsin(z)$. Using the provided function, we get that $J(\sin(\frac{1}{2}\arcsin(z))) = (c + 2)\arcsin(z)/2$ as required, so the provided function is a valid jumping distance function.
\end{proof}

\begin{corollary}\label{cor:interval}
If $\beta$ is in $I$, then $|\alpha - \beta| \leq (2 c + 2)\arcsin(z)$.
\end{corollary}

\subparagraph{Bounding the speed.}
To show that the orientation $\beta$ stays within the interval $I$, we argue that over a time period of $T(z)$ we can rotate $\beta$ at least as far as $E$. As the endpoint of the safe zone moves at most as fast as $E$, this implies that if $\beta$ leaves the safe zone at time $t$, it returns to it in the time period $(t, t + T(z)]$. Thus we require that $K T(z) \geq \Delta E(z, T(z))$, as $\beta$ can rotate at most $K$ units per time step when the points move at unit speed. We need to keep up only when the safe zone does not span all orientations, that is, the above inequality must hold only when $H(z) \leq \pi/2$ or $z \leq \sin(\frac{\pi}{2 c})$. For the following speed bound we choose a specific value $c = 3$.\footnote{We could have set $c = 3$ earlier and simplified some of the earlier analysis. We did not do so in order to demonstrate the general technique more clearly, rather than just specifically for this problem.} Hence we only need to chase $\alpha$ when $z \leq \sin(\frac{\pi}{6}) = \frac{1}{2}$.

In our proofs we use the following trigonometric inequalities.

\begin{lemma}\label{lem:trigineq}
The following inequalities hold for $0 \leq x \leq 1$:
\begin{enumerate}[noitemsep,topsep=0pt]
\item $\sin(\lambda \arcsin(x)) \leq \lambda x$ for $\lambda \geq 1$
\item $\sin(\lambda \arcsin(x)) \geq \lambda x$ for $0 < \lambda \leq 1$.
\end{enumerate}
\end{lemma}
\begin{proof}
We first show inequality (1). Let $x = \sin(y)$. We rewrite (1) into $\sin(\lambda y) \leq \lambda \sin(y)$. The derivative with respect of $y$ is $\lambda \cos(\lambda y)$ for the left side and $\lambda \cos(y)$ for the right side. Since $\cos(y) \geq \cos(\lambda y)$ for $0 \leq y \leq \pi/\lambda$ and $\lambda \geq 1$, we get that $\sin(\lambda y) \leq \lambda \sin(y)$ for $0 \leq y \leq \pi/\lambda$. In particular, for $y = \pi / (2 \lambda)$ we get that $1 = \sin(\lambda y) \leq \lambda \sin(y)$. Since $\sin(\lambda y) \leq 1$ and $\lambda \sin(y)$ attains it first maximum at $y = \pi/2$, we thus also get that $\sin(\lambda y) \leq \lambda \sin(y)$ for $0 \leq y \leq \pi/2$. Since $x = \sin(y)$ and $\sin(\pi/2) = 1$, the result follows. 

For inequality (2), set $x = \sin(\frac{1}{\lambda}\arcsin(y))$. We obtain that $y \geq \lambda \sin(\frac{1}{\lambda}\arcsin(y))$, or $\frac{1}{\lambda} \geq \sin(\frac{1}{\lambda}\arcsin(y))$. As shown above, this inequality holds for $0 \leq y \leq 1$. Since $y = \sin(\lambda \pi / 2) < 1$ implies $x = \sin(\frac{1}{\lambda}\arcsin(y)) = 1$, the inequality holds for $0 \leq x \leq 1$.
\end{proof}

\begin{lemma}\label{lem:arcsin}
$x \leq \arcsin(x) \leq \frac{\arcsin(a)}{a} x$ for $0 \leq a \leq 1$ and $0 \leq x \leq a$.
\end{lemma}
\begin{proof}
First note that $\arcsin(x)$ is a convex function for $0 \leq x \leq 1$. Since the derivative of $x$ and $\arcsin(x)$ is $1$ at $x = 0$, this directly implies that $x \leq \arcsin(x)$. Furthermore, since $\arcsin(x) = \frac{\arcsin(a)}{a} x$ for $x = 0$ and $x = a$, the convexity of $\arcsin(x)$ also directly implies the second inequality.
\end{proof}

\begin{lemma}\label{lem:speedbound}
If $K \geq 43$, then $|\beta(t) - \alpha(t)| \leq 8 \arcsin(z)$ (using $c = 3$) for all $t$.
\end{lemma}
\begin{proof}
Consider a time $t$ when $\beta(t)$ leaves the safe zone. We first argue that $\beta(t')$ will be in the safe zone at some time $t' \in (t, t + T(z)]$. To show this, we need to prove that $K T(z) \geq \Delta E(z, T(z))$ for $z \leq \frac{1}{2}$. 

To apply the bounds of Lemmata~\ref{lem:flipangle} and \ref{lem:flipaspectratio}, we must ensure that $T(z) = z/4$ satisfies the bounds for $\Delta t$. For Lemma~\ref{lem:flipangle}, observe that $(1 - z)/(2 + 2z)$ is decreasing and $z/4$ is increasing, and $(1 - z)/(2 + 2z) = \frac{1}{6} \geq z/4$ for $z = \frac{1}{2}$. For Lemma~\ref{lem:flipaspectratio} we apply Lemma~\ref{lem:trigineq} to show that $\sin(\frac{1}{2}\arcsin(z))/2 \geq z/4$. We thus get the provided bounds on $\Delta \alpha(z, T(z))$ and $\Delta z(z, T(z))$, and as a result a bound on $\Delta E(z, T(z))$. In particular, for $\Dt \leq T(z)$, we get:
\begin{align*}
\Delta E(z, \Dt) &\leq \Delta \alpha(z, \Dt) + H(z) - H(z - \Delta z(z, \Dt)) + J(z) - J(z - \Delta z(z, \Dt))\\
&= \arcsin(z + \Dt(2+2z)) + 8 \arcsin(z) - 8 \arcsin\left(\frac{\sin(\frac{1}{2}\arcsin(z)) - 2 \Dt}{1 + 2 \Dt}\right).
\end{align*}
We have that $z \leq \frac{1}{2}$, $\Dt \leq z/4 \leq \frac{1}{8}$ and $z + \Dt(2+2z) \leq \frac{7}{8}$. Then, using the inequalities of Lemma~\ref{lem:arcsin}, where $2\arcsin(\frac{1}{2}) \leq 1.05$ and $\frac{8}{7}\arcsin(\frac{7}{8}) \leq 1.22$, and using Lemma~\ref{lem:trigineq}, we get: 
\begin{align*}
\Delta E(z, \Dt) &\leq 1.22 (z + \Dt(2+2z)) + 8.4 z - \frac{4 z - 16 \Dt} {1 + 2 \Dt}\\
&\leq 9.62 z + 1.22 \Dt (2+2z),
\end{align*}
where the last inequality uses the fact that $\Dt \leq z/4$, and thus $16 \Dt \leq 4 z$. Finally, filling in $\Dt = T(z) = z/4$, we get:
\begin{align*}
\Delta E(z, T(z)) &\leq 10.23 z + 0.61 z^2 &\\
&\leq 10.6 z &\text{for } z \leq \frac{1}{2}\\
&\leq K \frac{z}{4} &\text{for } K \geq 43.
\end{align*}
Finally, we need to argue that $\beta(t)$ does not leave $I$ in the interval $(t, t + T(z)]$. To show this, we need to prove that $K \Delta t \geq \Delta E(z, \Delta t) - J(z)$ for all $\Delta t \in [0, T(z)]$. Using the inequalities above, we have:
\begin{align*}
\Delta E(z, \Dt) - J(z) \leq \arcsin(z + \Dt(2+2z)) + 3 &\arcsin(z) \\
&- 8 \arcsin\left(\frac{\sin(\frac{1}{2}\arcsin(z)) - 2 \Dt}{1 + 2 \Dt}\right).
\end{align*}
We first argue that this function is nondecreasing in $z$, such that $\Delta E(z, \Dt) - J(z) \leq \Delta E(\frac{1}{2}, \Dt) - J(\frac{1}{2})$. For that we consider its partial derivative in $z$:
\begin{align*}
\frac{\partial(\Delta E(z, \Dt) - J(z))}{\partial z} = &\,\frac{3}{\sqrt{1 - z^2}} + \frac{1+2\Dt}{\sqrt{1 - (2\Dt + z + 2\Dt z)^2}} - \\
&\frac{4 \cos(\frac{1}{2}\arcsin(z))}{\sqrt{1 - z^2} \sqrt{(1+2\Dt)^2 - (\sin(\frac{1}{2}\arcsin(z)) - 2\Dt)^2}}\\
\geq &\,\frac{3}{\sqrt{1 - z^2}} + \frac{1+2\Dt}{\sqrt{1 - z^2}} - \\
&\frac{4}{\sqrt{1 - z^2}}\frac{\cos(\frac{1}{2}\arcsin(z))}{\sqrt{1 - (\sin(\frac{1}{2}\arcsin(z)) - 2\Dt)^2}}\\
\geq &\,\frac{4}{\sqrt{1 - z^2}}\left(1-\frac{\cos(\frac{1}{2}\arcsin(z))}{\sqrt{1 - \sin^2(\frac{1}{2}\arcsin(z))}}\right)\\
\geq & \,0
\end{align*}
As a result we can conclude the following:
\begin{align*}
\Delta E(z, \Dt) - J(z) &\leq \Delta E\left(\frac{1}{2}, \Dt\right) - J\left(\frac{1}{2}\right)\\
&= \frac{\pi}{2} + \arcsin\left(\frac{1}{2} + 3\Dt\right) - 8 \arcsin\left(1 - \frac{4 + \sqrt{6} - \sqrt{2}}{4 + 8 \Dt}\right)
\end{align*}
Note that this bound is $0$ whenever $\Dt = 0$. It is now sufficient to show that the derivative of this function with respect to $\Dt$ is at most $K$ for $0 \leq \Dt \leq \frac{1}{8}$. Let $a = 4 + \sqrt{6} - \sqrt{2} \approx 5.035$.
\begin{align*}
\frac{\partial \Delta E(\frac{1}{2}, \Dt)}{\partial\Dt} &= \frac{6}{\sqrt{4 - (1 + 3 \Dt)^2}} + \frac{64 a}{(4 + 8\Dt)^2 \sqrt{\frac{2 a}{4 + 8\Dt} - (\frac{a}{4 + 8\Dt})^2}}\\
&= \frac{6}{\sqrt{4 - (1 + 3 \Dt)^2}} + \frac{64 a}{(4 + 8\Dt) \sqrt{2 a (4 + 8\Dt) - a^2}}\\
&\leq \frac{6}{\sqrt{4 - (1 + 3 \Dt)^2}} + \frac{16 a}{\sqrt{8 a - a^2}}\\
&\leq 4.14 + 20.86 \leq 43\leq K
\end{align*}
Here we used that $3/\sqrt{4 - (1 + 3 \Dt)^2}$ is increasing in $\Dt$ and that $\Dt \leq \frac{1}{8}$. We conclude that $K \Delta t \geq \Delta E(z, \Delta t) - J(z)$ for all $\Delta t \in [0, T(z)]$. Thus, $\beta(t)$ does not leave $I$ in the time period $(t, t + T(z)]$. By repeating this argument whenever $\beta(t)$ leaves the safe zone, we can conclude that $|\beta(t) - \alpha(t)| \leq 8 \arcsin(z)$ for all times $t$.
\end{proof}

\subsection{Lipschitz stability ratio}
What remains is to analyze the approximation ratio of the chasing algorithm for \obb. Corollary~\ref{cor:interval} implies that the orientation $\beta$ of the chasing algorithm is at most an angle $(2 c + 2)\arcsin(z)$ away from the orientation of the diameter.

\begin{lemma}\label{lem:obblipschitz}
If $|\beta - \alpha| \leq (2 c + 2) \arcsin(z)$, then $f_{\obb}(\beta, P) \leq (4 c + 6) \min_x f_{\obb}(x, P)$.
\end{lemma}
\begin{proof}


Assume that at some time $t$ we have a diametric box with diameter $D$ and aspect ratio $z$, and let $(p_1, p_2)$ be the diametrical pair. The smallest \obb must contain $p_1$ and $p_2$ and must hit the sides of the diametric box at, say, $q_1$ and $q_2$ (see Figure~\ref{fig:approx-obb} (left)). The smallest \obb must contain the triangles formed by $\{p_1, p_2, q_1\}$ and $\{p_1, p_2, q_2\}$; with the diametrical pair a fraction $x$ along the minor axis in the box, the height of these triangles with base $p_1p_2$ are $x \cdot Dz$ and $(1-x) \cdot Dz$ respectively. Their combined area is thus $D^2 z /2$ and provides a lower bound for the area of the diametric box.

\begin{figure}[b]
    \centering
    \includegraphics{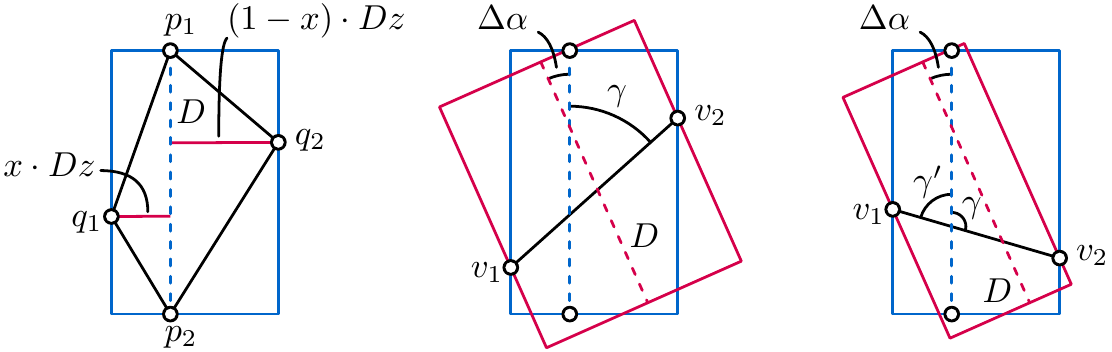}
    \caption{Illustrations supporting proof of Lemma~\ref{lem:obblipschitz}.
    }
    \label{fig:approx-obb}
\end{figure}
Now consider the box of the chasing algorithm, where $\Delta \alpha = |\beta - \alpha| \leq (2 c + 2)\arcsin(z)$. The major axis (in direction $\beta$) has length at most $D$. 
Let the minor axis be bounded by two points $v_1$ and $v_2$, and $\gamma$ the angle between the lines spanned by $v_1v_2$ and by the diametrical pair $p_1p_2$, on the opposite side of $\Delta \alpha$ with respect to $p_1p_2$. 
Let the smallest angle between those two lines be $\gamma'$. Then the distance between $v_1$ and $v_2$ is bounded by $z D / \sin(\gamma')$. 
Whenever $\gamma' = \gamma$, the angle between the minor axis of the chasing box and the line through $v_1$ and $v_2$ is $\pi/2 - \gamma - \Delta \alpha$ (see Figure~\ref{fig:approx-obb} (middle)). Thus, the length of the minor axis is $z D \cos(\pi/2 - \gamma - \Delta \alpha) / \sin(\gamma) = z D \sin(\gamma + \Delta \alpha) / \sin(\gamma)$.

However, it can also be the case that $\gamma' = \pi - \gamma$. The angle between the minor axis of the box and the line through $v_1$ and $v_2$ is now $\pi/2 - \gamma' + \Delta \alpha$ (see Figure~\ref{fig:approx-obb} (right)). Analogously, we hence find that the length of the minor axis is $z D \cos(\pi/2 - \gamma' + \Delta \alpha) / \sin(\gamma') = z D \sin(\gamma' - \Delta \alpha) / \sin(\gamma')$. Using $\gamma' = \pi - \gamma$, the length of the minor axis can be simplified to $z D \sin(\gamma + \Delta \alpha) / \sin(\gamma)$, which is the same expression as for $\gamma = \gamma'$.

Since the function $\sin(\gamma + \Delta \alpha) / \sin(\gamma)$ is decreasing in $\gamma$, we attain the maximum when $z / \sin(\gamma) = 1$ or $\gamma = \arcsin(z)$. Hence, using $\gamma = \arcsin(z)$, we get that the area of the box of the chasing algorithm is at most $D^2 \sin((2 c + 3)\arcsin(z))$, which is at most $D^2 z (2 c + 3)$ by Lemma~\ref{lem:trigineq}. Thus, $f_{\obb}(\beta, P) \leq (4 c + 6) min_x f_{\obb}(x, P)$.
\end{proof}


\begin{lemma}\label{lem:striplipschitz}
If $|\beta - \alpha| \leq (2 c + 2) \arcsin(z)$, then $f_{\strip}(\beta) \leq (4 c + 6) \min_x f_{\strip}(x)$.
\end{lemma}
\begin{proof}
Assume that at some time $t$ we have a diametric box with diameter $D$ and aspect ratio $z$, and let $(p_1, p_2)$ be the diametrical pair. The width of the diametric box is determined by two points $q_1$ and $q_2$.

We first derive a lower bound for the width of the thinnest \strip; note that this follows the same rationale as in the proof of Lemma~\ref{lem:flipaspectratio}.
Such a strip must contain the points $p_1$, $p_2$, $q_1$ and $q_2$. As adding points to a point set can only widen the thinnest strip, we consider just these four points for a lower bound.
For the thinnest \strip, all four points are on the boundary of the strip in the worst case, and we assume w.l.o.g. that $p_1$ and $q_1$ are on the same side of the strip (same for $p_2$ and $q_2$). Consider the following lines: $L$ oriented in the orientation of the strip (parallel to its boundary), $L_p$ spanned by $p_1p_2$ and $L_q$ spanned by $q_1q_2$. Let the angle between $L_p$ and $L_q$ be $\gamma$, $\gamma \geq \arcsin(z)$. 
The distance between $q_1$ and $q_2$ is then $z D / \sin(\gamma)$. 
We denote the angle between $L$ and $L_p$ by $\gamma_p$, and between $L$ and $L_q$ by $\gamma_q$.
We observe that $\gamma_p + \gamma_q = \gamma$, as the orientation of $L$ must bisect the $\gamma$ angle for the strip to be thinnest. 
The width of the strip is $\max(D \sin(\gamma_p), z D \sin(\gamma_q) / \sin(\gamma))$. We show that this width is at least $D \sin(\frac{1}{2}\arcsin(z))$. This is clearly the case if $\gamma_p \geq \frac{1}{2}\arcsin(z)$, so assume the contrary. Since the function $\sin(\gamma - \gamma_p) / \sin(\gamma)$ is increasing, it is optimal to set $\gamma = \arcsin(z)$. But then $z D \sin(\gamma_q) / \sin(\gamma) = D \sin(\gamma_q) > D \sin(\frac{1}{2}\arcsin(z))$. Thus, the width of the thinnest strip is at least $D \sin(\frac{1}{2}\arcsin(z)) \geq D z / 2$ by Lemma~\ref{lem:trigineq}.

Now consider the strip of the chasing algorithm, with an orientation $\beta$ differing at most $\Delta \alpha = |\beta - \alpha| \leq (2 c + 2)\arcsin(z)$ from the orientation of the diametrical pair. Let the width of the strip be bounded by two points $v_1$ and $v_2$, where the angle between the line through $v_1$ and $v_2$ and the line through the diametrical pair $p_{1}p_{2}$, opposite of $\Delta \alpha$ with respect to $p_{1}p_{2}$, is $\gamma$. Let the smallest angle between those two lines be $\gamma'$. Note that, the distance between $v_1$ and $v_2$ is $z D / \sin(\gamma')$, since they define the width of the diameter box. When $\gamma'=\gamma$, the angle between the vector perpendicular to the orientation of the strip and the line through $v_1$ and $v_2$ is $\pi/2 - \gamma - \Delta \alpha$. Thus, the width of the strip is $z D \cos(\pi/2 - \gamma - \Delta \alpha) / \sin(\gamma) = z D \sin(\gamma + \Delta \alpha) / \sin(\gamma)$. 

On the other hand, whenever $\gamma' = \pi - \gamma$, the angle between the vector perpendicular to $\beta$ and the line $v_{1}v_{2}$ is $\pi/2 - \gamma' + \Delta \alpha$. Analogously, we find that the width of the strip is $z D \cos(\pi/2 - \gamma' + \Delta \alpha) / \sin(\gamma') = z D \sin(\gamma' - \Delta \alpha) / \sin(\gamma')$. Using $\gamma' = \pi - \gamma$, the width of the strip can be simplified to $z D \sin(\gamma + \Delta \alpha) / \sin(\gamma)$, which is the same expression as for $\gamma' = \gamma$.

Since the function $\sin(\gamma + \Delta \alpha) / \sin(\gamma)$ is decreasing in $\gamma$, we attain the maximum when $z / \sin(\gamma) = 1$ or $\gamma = \arcsin(z)$. Thus, using $\gamma = \arcsin(z)$, the width of the strip is at most $D \sin((2 c + 3)\arcsin(z))$, which is at most $D z (2 c + 3)$ by Lemma~\ref{lem:trigineq}. We finally obtain that $f_{\strip}(\beta) \leq (4 c + 6) min_x f_{\strip}(x)$.
\end{proof}

By combining Lemmata~\ref{lem:speedbound},~\ref{lem:obblipschitz}, and~\ref{lem:striplipschitz}, we obtain the following bounds on the Lipschitz stability of \obb and \strip.

\begin{theorem}\label{thm:lipschitz}
The following Lipschitz stability ratios hold for \obb and \strip, assuming diameter $D(t) \geq 1$ for all $t$ and points move with at most unit speed:
\begin{itemize}[noitemsep,topsep=0pt]
    \item $\LS(\obb, 43) \leq 18$,
    \item $\LS(\strip, 43) \leq 18$. 
\end{itemize}
\end{theorem}

\section{Lipschitz stability of principal component}\label{sec:lipschitzpc}
The chasing algorithm does not work for the first principal component. Specifically, our scale normalization by requiring that the diameter is at least one at any point in time does not help. This can intuitively be attributed to the optimization function of \pc{}. Rather than being defined by some form of extremal points, $f_\pc{}$ is determined by variance: although the diameter may be large, many close points may still largely determine the first principal component. We formalize this via the lemma below. It implies that requiring a minimal diameter is not sufficient for a chasing algorithm with bounded speed to approximate \pc. The proof is inspired by the construction in \cite{DBLP:journals/comgeo/DimitrovKKR09} that shows the ratio on the areas between a bounding box aligned with the principal components and the optimal oriented bounding box can become infinite.

\begin{lemma}
For any constant $K$, there exists a point set $P(t)$ with minimum diameter $1$ at all times, such that any shape descriptor that approximates the optimum of $f_{\pc}$ must move with speed strictly greater than $K$.
\end{lemma}
\begin{proof}[Proof sketch]
Consider the point set $P$ containing two points that lie on a horizontal line and form a diametrical pair for $P$. All other points in $P$ form a dense subset $P'$ that is not collinear with the two points that form the diametrical pair, but are located very close to only one of the two points. See Figure~\ref{fig:pca-diam} for the construction.

Remember that the optimization function for the first principal component minimized the sum of squared distances from the points to the line.
The dense subset $P'$ contains so many points that any line that differs more than $\epsilon$ from the orientation of the line $l$ through $P'$, has a significantly larger sum of squared distances from the points in $P$ to $l$.
Hence \pc{} follows $P'$ regardless of the position of the two points forming the diametrical pair.

\begin{figure}[ht]
    \centering
    \includegraphics{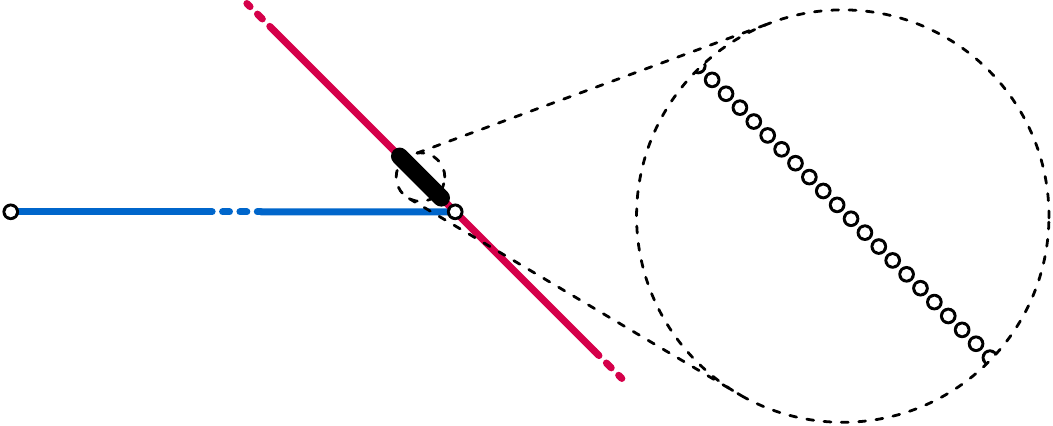}
    \caption{Construction that shows how \pc{} can move arbitrarily fast, despite a minimal diameter. The two points connected by the blue line form a diametrical pair for the whole point set. The dense set of points located arbitrarily close to the right point can move around it in an infinitesimally short amount of time. Because the point set is so dense, the orientation of the first principal component (in red) follows it regardless of how far the leftmost point is placed.}
    \label{fig:pca-diam}
\end{figure}

For any constant $K$, the points in $P'$ can be placed and moved in such a way that in an infinitesimally small time frame, they can move around one of the points of the diametrical pair and change the orientation of \pc{} by more than $K$.
Thus any shape descriptor that approximates the optimum of $f_{\pc}$ must also change its orientation by more than $K$.
\end{proof}

\section{Conclusion}
\label{sec:conclusion}

We studied the topological stability (ratio of continuous solutions to optimal discontinuous solutions) and Lipschitz stability (ratio of continuous solutions with bounded speed to optimal discontinuous solutions) of three common orientation-based shape descriptors. Although stateless algorithms cannot achieve topological stability, we proved tight bounds on the topological stability ratio for state-aware algorithms.
Our Lipschitz analysis focuses on upper bounds, showing that a chasing algorithm achieves a constant ratio for a constant maximum speed, for two of the three considered descriptors. It remains open to establish whether lower bounds exist that are stronger than those already given by our topological stability results.
We believe that our analysis techniques for the Lipschitz upper bounds are of independent interest, to analyze other problems that could be approached via a chasing algorithm.

\bibliography{references}

\end{document}